\documentclass[journal, onecolumn]{IEEEtran}
\usepackage[final]{graphicx}
\usepackage[reqno]{amsmath}
\usepackage{amssymb}
\usepackage{cite}
\usepackage{enumerate}
\usepackage{algorithmicx}
\usepackage{algorithm}
\usepackage[normalem]{ulem}
\usepackage{psfrag}
\usepackage{subfigure}
\usepackage{times,epsfig}
\usepackage{latexsym}
\usepackage{graphicx,amstext}
\usepackage{setspace}
\usepackage{eucal}
\usepackage{mathrsfs}
\usepackage{color}
\usepackage{url}

\usepackage{amsmath}
\usepackage{amssymb,amsthm}
\usepackage{lipsum, colortbl}
\usepackage[utf8]{inputenc}

\usepackage{algpseudocode}
\usepackage[left=0.65in,right=0.65in,bottom=0.95in,top=0.7in]{geometry}
\definecolor{Gray}{gray}{0.9}
\newcolumntype{s}{>{\columncolor[HTML]{AAACED}} p{1.75cm}}

\usepackage{booktabs,multirow,bigstrut,url}

\def\eqdef{\triangleq}
\newtheorem{theorem}{Theorem}
\newtheorem{lemma}{Lemma}%[chapter]
\newtheorem{remark}{Remark}%[chapter]

\begin{document}

\title{MIMO with Energy Recycling}
% \author{Michael~Shell,
%         John~Doe,~\IEEEmembership{Fellow,~OSA,}
%         and~Jane~Doe,~\IEEEmembership{Life~Fellow,~IEEE}% <-this % stops
% qa space
% , Frank Barickman$^\dagger$, John Martin$^\dagger$, Josh Weston$^\dagger$
\author{\IEEEauthorblockN{Y. Ozan Basciftci, Amr Abdelaziz and C. Emre Koksal}
\IEEEauthorblockA{ \\ $\{$basciftci.1, abdelaziz.7, koksal.2$\}$@osu.edu\\
Department of Electrical and Computer Engineering \\
The Ohio State University\\
Columbus, Ohio 43201
}}


\maketitle

\begin{abstract}
We consider a Multiple Input Single Output (MISO) point-to-point communication system in which the transmitter is designed such that, each antenna can transmit information or harvest energy at any given point in time. We evaluate the achievable rate by such an energy-recycling MISO system under an average transmission power constraint. Our achievable scheme carefully switches the mode of the antennas between transmission and wireless harvesting, where most of the harvesting happens from the neighboring antennas' transmissions, i.e., recycling. We show that, with recycling, it is possible to exceed the capacity of the classical non-harvesting counterpart. As the complexity of the achievable algorithm is exponential with the number of antennas, we also provide an almost linear algorithm that has a minimal degradation in achievable rate. To address the major questions on the capability of recycling and the impacts of antenna coupling, we also develop a hardware setup and experimental results for a 4-antenna transmitter, based on a uniform linear array (ULA). We demonstrate that the loss in the rate due to antenna coupling can be made negligible with sufficient antenna spacing and provide hardware measurements for the power recycled from the transmitting antennas and the power received at the target receiver, taken simultaneously. We provide refined performance measurement results, based on our actual measurements.
\end{abstract}


% \fi
\IEEEpeerreviewmaketitle
\allowdisplaybreaks
\section{Introduction}
When information is transmitted wirelessly, only a small fraction of the emitted power reaches the intended receive antenna. Energy recycling is based on the premise that some of the emitted energy can be captured back and reused by the transmitter itself. In principle, a wireless transmitter equipped with energy harvesting capabilities may benefit, not only from other natural and man-made sources of wireless energy \cite{ulukus2015energy}, but also from its own transmitted power. Indeed, if done properly, this provides a significant improvement over wireless energy transfer, since the power is captured from a nearby antenna, where the power received can be orders of magnitude higher than the power that can be received from a far away station. 

Energy recycling is particularly a good fit for Multi Input Multi Output (MIMO) communication due to the potentially large number of antennas to harvest from. Indeed, with transmit beamforming, due to spatial diversity, the transmission power over each antenna can vary substantially. Therefore, at any given time, the contribution of different transmission antennas to the achieved rate can be different. If the transmit antennas can switch between two modes, signal transmission and energy harvesting, we can use antennas with relatively low contribution at a time to energy harvesting mode. The hypothesis is that, by doing so, the savings in average power can make up for the minimal loss due to not transmitting over that antenna at that time. In this paper, we show that this is indeed the case and with a careful control of switching between harvesting and transmission modes, for a given average transmission power constraint, the achievable rate can be higher than the capacity of the classical MIMO system without energy recycling. Unlike most of the available literature, we are concerned with the opportunity of energy recycling at the transmitter side rather than harvesting it at the receiver side. 
%In particular, in a point-to-point communication scenario with multi-antenna transmitter, we show that, even if a fraction of the transmitted power can be recycled, it has a significant impact on the achievable communication rate. Indeed, in typical communication settings there are tens of dB difference between power observed at the transmitter itself and the power that reaches the receiver.

%An important issue in simultaneous wireless information and power transfer (SWIPT) is the following: while omnidirectional antennas serve the coverage requirements for mobile communications, efficiency of energy transfer suffers due to large path losses associated with large beamwidths. Indeed, in typical situations, the received power is typically tens of dB lower than the transmitted power (this is the major cause for the self-interference problem in full-duplex communication). We show that, even if part of this generated signal can be recycled, it has a much more significant impact on the network performance than the energy harvested from the power transferred from other nodes. 

We study a Multi Input Single Output (MISO) point-to-point communication model in which, at any given point in time, the transmitter deactivates a subset of its transmitting antennas and assigns them for energy harvesting.
%\footnote{Throughout this paper, the set of transmitting and harvesting antennas are referred to as the active antennas and harvesting antennas, respectively.}. 
The energy harvested over these antennas is then recycled and can be reused by the transmitter, see Figure \eqref{fig:sys_mdl}. We study the information-theoretic limits of the aforementioned harvesting MISO system in terms of its achievable communication rate. We show that the rate achievable by the energy-recycling MISO system is above the capacity of the conventional MISO system without recycling, with an increasing margin as the number of antennas increase. Our achievable scheme includes an antenna scheduling algorithm that determines which antennas should be active as a function of the antenna-to-antenna loss gains at the transmitter and the transmitter-to-receiver channel gains. The complexity of the developed optimal antenna scheduling algorithm is found to grow exponentially with the number, $M$, of the transmitter antennas. To address this problem, we provide $O(M \log M)$ complexity algorithm, which achieves the optimal performance in a symmetric system, where the propagation loss between each pair of transmission antennas is identical. Although we considered a MISO channel for the sake of simplicity of the exposure of the ideas, the analysis conducted can be generalized to a MIMO setting in a straightforward fashion.

One of the main questions around energy recycling is its practicality in real situations. The major problem here is, due to antenna coupling, the harvesting antennas affect the energy received by the receive antennas in the far field. To understand the severity of this problem, we conduct hardware experiments using universal serial radio peripheral (USRP) devices. Clearly, the closer the harvesting antenna to an active one, the higher the harvested energy. While understanding this behavior, our experiments reveal the amount of decrease in the power at the target receiver.
%Further, we provide hardware experimental results using universal serial radio peripheral (USRP) devices. Experimental results show that, the closer the harvesting antenna to an active one, the higher the harvested energy but also the less the power received at the target receiver due to antenna coupling. 
We experimentally identify the optimal tradeoff between the harvested and the received power for uniform linear arrays (ULA) at the transmitter. We show that, the best recycling-receive power tradeoff is observed when the antenna-spacing in the array is identical to a quarter wavelength. 

\begin{figure}[t]
   \centering
   \includegraphics[width=0.350\textwidth]{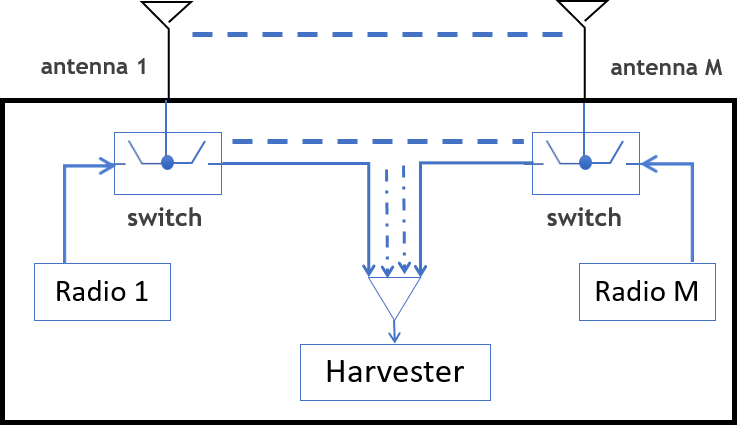}
   \caption{Transmitter with $M$ antennas. The switches control the role of each antenna, as they can actively transmit the signal fed by the radio or harvest ambient energy, a large portion of which comes from actively-transmitting antennas.}
   \label{fig:sys_mdl}
 \end{figure}   

Several research efforts are made to develop communication schemes for networks composed of energy harvesting nodes, equipped with a single antenna. The fundamental tradeoff between the rates at which energy and reliable information can be transmitted over a single noisy line is studied in \cite{varshney2008transporting}. In \cite{grover2010shannon}, the problem of wireless information and power transfer across a noisy coupled-inductor circuit, which is a frequency-selective channel with additive white Gaussian noise. The optimal tradeoff between the achievable rate and the power transferred is characterized given the total power available. Several other works considering, mainly, the optimal control policy for networks with energy harvesting nodes can be found in the literature, see e.g. \cite{gatzianas2010control,sharma2010optimal,yang2012optimal,ozel2011transmission,gunduz2011two, huang2013throughput, varan2013energy, vaze2013transmission, srivastava2013basic, jog2014energy}. Simultaneous energy and information transfer literature has been extended to MIMO broadcast \cite{zhang2013mimo}, fading \cite{liu2013wireless} and interference \cite{park2013joint} channels. 

To our best knowledge, this is the first study that considers the basic limits of wireless communication with energy recycling. Along with providing efficient resource allocation schemes with almost capacity-achieving performance, the other unique feature of our study is that, we provide a thorough evaluation of the concept of energy recycling using an actual hardware setup with USRPs.
\section{System Model}
We consider a point-to-point MISO communication in which a transmitter equipped with $M$ antennas wishes to send a message to a receiver with a single antenna. In our system, at any given point in time, a selected subset of the antennas may not transmit. Instead they harvest wireless energy, a substantial portion of which is coming from the active antennas. We call this communication model \textbf{recycling MISO}. This is unlike the traditional \textbf{non-recycling} MISO communication~\cite{telatar1999capacity,tse2005fundamentals}, where each antenna can only be in the transmission mode. In our system, the antennas that transmit are referred to as {\em active} and the antennas that remain silent and harvest energy are referred to as {\em harvesting}. 

Let $\mathcal A$ be the set that denotes indices of the active antennas at the transmitter at a particular channel use. Then, the observed signal at harvesting antenna $l$ of the transmitter at that channel use can be written as 
\begin{equation}
Z_l = \sum_{k\in\mathcal A} \sqrt{\alpha_{kl}}X_k  +N_l, \label{l1}
\end{equation}
where $X_k$ is the complex transmitted signal at active antenna $k$, $\alpha_{kl}$ is the path loss of the channel connecting $k$-th antenna to $l$-th antenna (both at the transmitter), and $N_l$ is the additive noise signal, distributed as $\mathcal{CN}(0,1)$. %Note that $X=\{ X_k\}_{k\in\mathcal N}$.
The corresponding harvested energy for $Z_k$ is $|Z_k|^2$. Note that $\{\alpha_{kl}\}_{k,l\in \mathcal M}$ depend on the geometry of antenna placements at the transmitter, where $\mathcal M \eqdef\{1,\ldots, M\}$. As seen in \eqref{l1}, in the recycling MISO communication, passive antennas harvest energy using the incoming signal from their active neighboring antennas.

We assume a flat fading transmitter-to-receiver channel, in which the channel gains in different channel uses are identically and independently distributed (i.e., rich scattering is present between the transmitter and the receiver). This assumption is merely for the simplicity of the analysis, but generalization to other propagation settings is straightforward. The observed signal at the single receive antenna at a particular channel use can be written as
\begin{align}
Y_k = \sum_{k\in\mathcal A} X_k G_k+W , \label{sys_mod}
\end{align}
where $Y_k$ is the received complex signal and $W$ is the additive noise signal, distributed as $\mathcal{CN}(0,1)$.  In the above equation, $G_k$ is the complex gain of the channel connecting the $k$-th antenna at transmitter to the receiver. We assume that the channel gains $G\eqdef [G_1,\ldots G_M]$ are independent through space and are distributed as $\mathcal{CN}(0,I_M)$, where $I_M$ is $M\times M$ identity matrix and parameter $d$ denotes the path loss gain of the channel between the transmitter and the receiver.% For instance, when $d=10^{-6}$, the path loss is $-60dB$. 

In this paper, we consider that full channel state information (CSI) is available, where both transmitter and the receiver have the perfect information of the channel gains. Hence, the transmitter has the ability to choose active and passive antennas as well as the transmission power as a function of the instantaneous channel gains. We develop antenna scheduling algorithm that determines which antennas should be active as a function of the antenna-to-antenna loss gains $\{\alpha_{kl}\}_{k,l=1,\ldots,M}$ at the transmitter and the channel gains. For the active antennas, we evaluate the capacity-achieving power allocation. Ultimately, we compare the achievable rates with and without recycling, subject to an average energy constraint.

%We consider the performance of the MISO communication under two cases 1) No-recycling: All the antennas are active and transmit signal, 2) recycling: Some antennas keep silent intentionally in order harvest energy using the incoming signals coming from active neighbor antennas.  We next present the capacity of no-recycling case, which is well-known in the literature.  
\section{Basic Limits of Energy-Recycling MISO}
In this section, we evaluate a rate that is achievable with our energy-recycling MISO communication system. We also provide the associated antenna scheduling and power allocation schemes that maximize this rate. 

\subsection{Basic Limits}

First, we state the capacity of the classical no-recycling MISO communication~\cite{tse2005fundamentals,telatar1999capacity} with the channel model we presented in the previous section:
\begin{align}
C_{\text{no-ryc}} = &\max_{P(\cdot)} \mathbb E\left[\log\left(1+P(H)\sum_{k=1}^M H_k\right)\right]\label{eq:no-harvest-rate}\\
&\text{subject to}\ \mathbb E \left[P(H)\right]\leq P_c
\end{align}
where $P_c$ is  positive real number that denotes the average power constraint, $H_k \eqdef |G_k|^2$ denotes the power gain of the channel connecting $k$-th antenna at the transmitter to the receiver, and
$H\eqdef \left[H_1,\ldots, H_M\right]$. In~\eqref{eq:no-harvest-rate}, $P(\cdot):\mathbb R_{+}^M\to \mathbb R_{+}$ denotes the power allocation function, where $\mathbb R_+$ denotes all non-negative real numbers. Note that, here the average power constraint is a long-term average across time, summed over all transmission antennas.

The associated capacity-achieving power allocation function can be found to be the water-filling solution:
\begin{equation}
\label{eq:waterfilling}
P(h) = \left(\lambda-\frac{1}{\sum_{k=1}^M h_k}\right)^+,
\end{equation}
where $h\eqdef [h_1,\ldots, h_M]$ is the channel power gain sequence and $\lambda$ is non-negative real number that is chosen such that $\mathbb E\left[P(H)\right] = P_c$. Note that, Eq.~(\ref{eq:waterfilling}) provides the total power allocation {\em over all antennas} at a given point, $k$, in time. The distribution of power {\em across antennas} to achieve the capacity given in (\ref{eq:no-harvest-rate}) is found to be {\em conjugate beamforming}, which leads to the following transmitted signal vector across the antennas:
\begin{equation}
\mathbf{X} = \frac{G^*}{||G||}\cdot s,
\nonumber
\end{equation}
where $s$ is the complex data signal.

%\begin{remark}

%\end{remark}

Our first theorem provides a rate achievable with energy-recycling MISO communication. Subsequently, we compare the achievable rate with the capacity, \eqref{eq:no-harvest-rate}, of the classical non-recycling system.
\begin{theorem}\label{lemma:mainresult} The following rate is achievable with MISO harvesting communication:
\begin{align}
R_{\text{ryc}} = &\max_{\mathcal A(\cdot),P(\cdot)} \mathbb E\left[\log \left(1 + P(H)\sum_{k\in \mathcal A(H)}H_k\right)\right]\label{eq:harvesting_rate}\\
&\text{subject to}\ \mathbb E\left[P(H)f(H)\right]\leq P_c,
\label{eq:harvesting_power_constraint}
\end{align}
where $\mathcal A(\cdot)$ is a mapping from $M$ dimensional non-negative real space $\mathbb R_{+}^M$ to a non-empty subset of $\mathcal M$ and function $f(\cdot):\mathbb R_{+}^M\to \mathbb R_{+}$ is identical to:
$$
f(h) = \left(1-\frac{\sum_{k\in\mathcal{A}(h)}\sum_{l\in\mathcal{A}^c(h)}h_k\alpha_{kl}}{\sum_{k\in\mathcal{A}(h)}h_k}\right).
$$ for any $h\in \mathbb R_{+}^M$.\qed
\end{theorem}
The proof of Theorem~\ref{lemma:mainresult} can be found in Appendix~\ref{sec:lemma_proof}.
In Theorem~\ref{lemma:mainresult}, subset $\mathcal A(h)\subseteq \mathcal M$ denotes the set of active antennas for channel power sequence $h$. The transmitter keeps the antennas in subset $\mathcal A^c(h)$ in the harvesting mode in order to harvest energy using the signal coming from the active antennas in set $ \mathcal A(h)$. In order to achieve rate~\eqref{eq:harvesting_rate}, the transmitter employs conjugate beamforming across the active antennas, similar to the non-recycling setting. Also, observe that, the average power constraint is on the consumed energy rather than the transmitted energy. For example, suppose the average power constraint is $1$mW, we allow our transmitters to transmit an average power of $1.1$mW if an average power of $0.1$mW is recycled back into the batteries.

%while $P_c$ is the average transmission power constraint, we assume that our proposed harvesting system is allowed to utilize the harvested energy to transmit with an average power that may, under certain conditions, exceed $P_c$. An alternative strategy is to keep the transmit power not exceeding $P_c$ while benefiting from the harvested energy in the form of energy saving. The latter strategy is beyond the scope of this paper. 

%We next provide a couple of remarks. In the first remark, we simplify the power constraint in~\ref{} when $\alpha_{kl}$'s are identical for all $k,l\in\mathcal M$. 

%In the following remark, we compare the rate achieved with and without energy recycling.
%\begin{remark}Suppose that $\alpha_{kl}$ is identical  and equal to non negative real number $\alpha$ for any $k,l\in\mathcal M$ and $k\neq l$. Then, the left hand side (LHS) of the average power constraint becomes 
%\begin{align}
%&\text{LHS of \eqref{eq:harvesting_power_constraint}}=\mathbb E\left[P(H)(1-|\mathcal{A}^c(H)|\alpha)\right]\\
%&=\mathbb E\left[P(H)(1-|\mathcal A^c(H)|\alpha)\right]\\
%&=(1-|\mathcal N^c_s|\alpha)\mathbb E\left[P(H)\right]. \label{power_constraint}
%\end{align}
%Hence, the power constraint in~\eqref{eq:harvesting_power_constraint} reduces to the following constraint:
%\begin{equation}
%\mathbb E\left[P(H)\right]\leq \frac{P_c}{1-|\mathcal N^c_s|\alpha}.
%\end{equation}\qed
%\end{remark}
\begin{remark} One can observe that $R_{\text{ryc}} \geq C_{\text{no-ryc}}$: Let $\mathcal A(h)=\mathcal M $ for any $h\in \mathbb R_{+}^M$. Then, $f(h)=1$ for any $h$ and the optimization problem in~\eqref{eq:harvesting_rate} becomes equal to that in~\eqref{eq:no-harvest-rate}.\qed
\end{remark}

\subsection{Optimal Antenna Scheduling and Power Allocation}

We next find a power allocation function $P(\cdot)$ and mapping $\mathcal A(\cdot)$ solving the optimization problem in~\eqref{eq:harvesting_rate}. Define power allocation function $P_{\text{new}}(\cdot):\mathbb R_+^M\to \mathbb R_+$ as 
\begin{equation}
P_{\text{new}}(h) \eqdef P(h)f(h)\nonumber
\end{equation}
for any $h\in \mathbb R_+^M$. We can write the optimization problem in~\eqref{eq:harvesting_rate} in terms of $P_{new}(\cdot)$ as 
\begin{align}\label{eq:harvesting_rate2}
\begin{aligned}
R_{\text{ryc}} =& \max_{\mathcal A(\cdot),P(\cdot)} \mathbb E\left[\log \left(1 + P_{\text{new}}(H)g_{\mathcal A}(H)\right)\right]\\
&\text{subject to}\ \mathbb E\left[P_{\text{new}}(H)\right]\leq P_c
\end{aligned}
\end{align}
%due to the fact that $P(h)I_{h\in \mathcal A_s} = P_{new}(h)I_{h\in \mathcal A_s}$ for any $s$
where function $g_{\mathcal A}(\cdot):\mathbb R_{+}^M\to \mathbb R_{+}$ is defined as
\begin{align*}
g_{\mathcal A}(h) \eqdef \frac{\sum_{k\in\mathcal A(h)} h_k}{f_{\mathcal A}(h)}
=\frac{(\sum_{k\in \mathcal{A}(h)}h_k)^2}{\sum_{k\in \mathcal{A}(h)}h_k-\sum_{k\in \mathcal{A}(h)}\sum_{k\in\mathcal{A}^c(h)}h_k\alpha_{kl}}
\end{align*}
for any mapping $\mathcal A(\cdot):\mathbb R_{+}^M\to \mathbb M $. Notice that, in~\eqref{eq:harvesting_rate2}, the maximization is over $P(\cdot)$. Hence replacing $P(\cdot)$ with $P_{\text{new}}(\cdot)$ does not change the value of the optimization problem. We can rewrite the optimization problem in~\eqref{eq:harvesting_rate2} as  
\begin{align}\label{eq:harvesting_rate4}
\begin{aligned}
R =& \max_{\mathcal A(\cdot),P_{\text{new}}(\cdot)} \mathbb E\left[\log \left(1 + P_{\text{new}}(H)g_{\mathcal A}(H)\right)\right]\\
&\text{subject to}\ \mathbb E\left[P_{\text{new}}(H)\right]\leq P_c.
\end{aligned}
\end{align}
%After rewriting the optimization problem in~\eqref{eq:harvesting_rate} as in~\eqref{eq:harvesting_rate4},  we are ready to provide a mapping $\mathcal A(\cdot)$ that solves \eqref{eq:harvesting_rate4}.  
We next provide mapping $\mathcal A(\cdot)$ that solves \eqref{eq:harvesting_rate4}. First note that for any power allocation function $P_{\text{new}}(\cdot)$ and mapping  $\mathcal A(\cdot):\mathbb R_{+}^M\to \mathbb M $, inner maximization in~\eqref{eq:harvesting_rate4} is bounded as
\begin{multline}
\label{eq:find_subsets1}
\mathbb E\left[\log \left(1 + P_{\text{new}}(H)g_{\mathcal A}(H)\right)\right] \\
\leq \mathbb E\left[\log \left(1 + P_{\text{new}}(H)g_{\text{max}}(H)\right)\right],
\end{multline}
where $g_{max}(H) \eqdef \max_{\mathcal A} g_{\mathcal A}(H)$. Hence, mapping $\mathcal A_{\text{max}}(\cdot)$, defined as 
\begin{equation}
\mathcal A_{\text{max}}(h) \eqdef \arg\max _{\mathcal{A(\cdot)}} g_{\mathcal A}(h)\nonumber
\end{equation}
for all $h\in \mathbb R_+$ solves the optimization problem in~\eqref{eq:harvesting_rate4}. We conclude that $A_{\text max}(h)$ is the set of antennas that the transmitter schedules as active for channel power sequence $h$ in order to achieve the rate proposed in Theorem~\ref{lemma:mainresult}.

Replacing the inner maximization in~\eqref{eq:harvesting_rate4} with the expectation term in (\ref{eq:find_subsets1}) leads to the following optimization problem:
\begin{align}\label{eq:harvesting_rate3}
\begin{aligned}
R_{ryc} =&\max_{P_{\text{\text{new}}}(\cdot)} \mathbb E\left[\log \left(1 + P_{\text{\text{new}}}(H)g_{\text{max}}(H)\right)\right]\\
&\text{subject to}\ \mathbb E\left[P_{\text{new}}(H)\right]\leq P_c.\\
\end{aligned}
\end{align}
The power allocation function that solves the optimization problem in~\eqref{eq:harvesting_rate3} has the water-filling form, given by:
\begin{equation}
P_{\text{\text{new}}}(h) = \left(\lambda-\frac{1}{g_{\text{max}}(h)}\right)^+,\nonumber
\end{equation}
where $\lambda$ is non-negative real number that is chosen such that $\mathbb E\left[P_{\text{new}}(H)\right] = P_c$.

\subsection{Low-complexity Near-Optimal Antenna Scheduling and Power Allocation}

Note that in order to find the active antennas at given channel power gain vector $h$, the transmitter needs to find $\mathcal A_{\text {max}}(h)$. Since there are $2^M-1$ non-zero subsets of $\mathcal M$, the antenna scheduling process will take a time that varies exponentially with the number of antennas, $M$. %The reason is the transmitter has to $2^M-1$ comparisons in order to find the maximum 
%and hence time complexity will be $O(2^M)$

To address this complexity issue, we next provide an algorithm for computing the active antenna set, $\mathcal A_{\text max}(h)$ for a given channel power sequence $h=[h_1, \ldots, h_M]$, with a time complexity that scales as $O(M\log M)$. We also show that the low-complexity antenna scheduler is optimal when all coefficients $\{\alpha_{kl}\}_{k,l=1,\ldots M}$ are identical.

Suppose that $\alpha_{kl} = \alpha$ for any $k,l\in\mathcal M$. Mapping $g_{\mathcal{A}}(h)$ reduces to:
\begin{align}
g_{\mathcal A}(h)&=\frac{(\sum_{k\in \mathcal{A}(h)}h_k)^2}{\sum_{k\in \mathcal{A}(h)}h_k-\sum_{k\in \mathcal{A}(h)}\sum_{l\in\mathcal A (h)^c}h_k\alpha} \nonumber \\
&=\frac{\sum_{k\in \mathcal A(h)} h_k}{1-(M-|\mathcal A(h)|)\alpha}. \label{simplified}
\end{align}
Suppose that $\{h_{n_k}\}_{k\in\{1,\dots,M\}}$ is a channel power sequence sorted in a descending order such that $h_{n_{k}}\geq h_{n_{k+1}}$ for all $k\in{1,\ldots M-1}$.
We define function $g_i(\cdot):\mathbb R_+^M\to \mathbb R_+$ as 
$$
g_i(h) \eqdef\frac{\sum_{k=1}^i h_{n_k}}{1-(M-i)\alpha}.
$$
Note that $g_i(h)=g_{\mathcal A}(h)$ for $\mathcal A(h) = \{n_1, \dots n_i\}$ due to definition above and \eqref{simplified}. Further,
\begin{align}
g_i(h) = \max_{\{ \mathcal A(\cdot)\ \left| \ |\mathcal A(h)|=i \right. \}} g_{\mathcal A}(h)\nonumber
\end{align}
due to the fact that $\{h_{n_1}, h_{n_2},\ldots, h_{n_i}\}$ are the $i$ largest channel power gains in $h$. Then, we can write the following:
\begin{align}
 \max_{i\in\mathcal M}g_i(h) &=\max_{i\in\mathcal M}\max_{\{ \mathcal A(\cdot)\ \left| \ |\mathcal A(h)|=i \right. \}} g_{\mathcal A}(h) \nonumber\\
&=\max_{\mathcal A(\cdot)}g_{\mathcal A}(h). \nonumber
\end{align}
Defining $i^{*}\eqdef \arg \max_{i\in\mathcal M}g_i(h)$, we conclude that $\{n_1, \dots n_{i^{*}}\} =  \mathcal A_{\text max}(h)$. Hence, once the transmitter computes index $i^*$, it finds which antennas to keep active. Exploiting the fact that all $\alpha_{kl}$'s are identical, we converted the problem $\arg\max_{\mathcal A(\cdot)} g_{\mathcal A}(h)$ into a simpler problem $\arg\max_i g_i(h)$. The time complexity of solving the first problem is $O(2^M)$ due to the fact there are $2^M-1$ non-empty subset of $\mathcal M$. On the other hand, solving the second problem has a time complexity of $O(M)$. 

Furthermore, the transmitter does not even need to compute all $g_i(h)$'s in order to find index $i^*$. Next lemma clarifies this proposition.

\begin{lemma}\label{lemma:max}
For any $h\in \mathbb R_+^M$, if there exists $i\in \{1,\ldots, M-1\}$ such that $g_{i}(h)\geq g_{i+1}(h)$, then $g_{i}(h)\geq g_{j}(h)$ for any $j$ such that $i\leq j\leq M$.
\end{lemma}
\begin{proof}
Suppose there exists $i\in \{1,\ldots, M-1\}$ such that $g_{i}(h)\geq g_{i+1}(h)$. The inequality  $g_{i}(h)\geq g_{i+1}(h)$ is equivalent to the following:
\begin{align}
\alpha \sum_{k=1}^i h_{n_k}\geq h_{n_{i+1}}\cdot(1-(M-i)\alpha)\label{given}
\end{align}
Pick an arbitrary integer $t$ such that $i+t\leq M$. Similarly, the inequality  $g_{i}(h)\geq g_{i+t}(h)$ is equivalent to the following:
\begin{align}
t\cdot\alpha \sum_{k=1}^i h_{n_k}\geq \sum_{k=i+1}^{i+t} h_{n_{k}}\cdot(1-(M-i)\alpha)\label{tobeproved}
\end{align}

In order to prove Lemma~\ref{lemma:max}, we need to show that the inequality in~\eqref{given} implies the inequality in~\eqref{tobeproved}. Hence, assuming the inequality in~\eqref{given} is true, we have the following:
\begin{align}
t\cdot \alpha \sum_{k=1}^i h_{n_k}&\geq  t\cdot h_{n_{i+1}}(1-(M-i)\alpha)\nonumber\\
&\geq \sum_{k=i+1}^{i+t}h_{n_{k}}\cdot(1-(M-i)\alpha),\nonumber
\end{align}
where the first inequality follows from the inequality in~\eqref{given}. The second inequality follows from the fact that $h_{n_{i+1}}\geq \max_{j:i+1\leq j\leq i+1} h_{n_j}$.
\end{proof}

 We now provide the antenna scheduling algorithm. First the transmitter sorts the channel power sequence  $\{h_{n_k}\}_{k\in\{1,\dots,M\}}$ in a descending order such that $h_{n_{k}}\geq h_{n_{k+1}}$ for all $k\in{1,\ldots M-1}$. Then, the transmitter evaluates $i^{*}= \arg \max_{i\in\mathcal M}g_i(h)$ using Lemma~\ref{lemma:max}. More specifically, starting at $i=1$ and increasing $i$ by one at each step, the transmitter computes $g_i(h) $ until step $j$ where $g_j(h)\geq g_{j+1}(h)$. The transmitter sets $i^*$ to $j$. If there is no $i\in \{1,\ldots, M-1\}$ such that $g_{i}(h)\geq g_{i+1}(h)$, the transmitter sets $i^*$ to $M$. Finally, transmitter schedules antennas with indices $\{n_1,\ldots, n_{i^*}\}$ as active antennas. 
 
 The steps of the algorithm are formally given below:
 
%Due to the lemma, antennas subset is equal to $\mathcal A_{\text{max}}(h)$.
 %%%%%%%%%%%%%%%%%%%%%%%%%%%%%%%%%%%%%%%%%%%%%%%%%%%%%%%%%%%%%%%%%%%%%%%%%%%%%%%%%%%

\begin{algorithm}
\label{alg:ant_sc}
   \caption{Antenna Scheduling}
    \begin{algorithmic}[1]
      \State Sort {$\{h_{n_k}\}_{k\in\{1,\dots,M\}}$}\Comment{such that $h_{n_{k}}\geq h_{n_{k+1}}$ for all $k\in{1,\ldots M-1}$}
        \State $i=1$
%        \State Evaluate {$g_1(h)$}
%        \For{$i=2$ to $M$} 
%                \State Evaluate {$g_i(h)$}
%                \If{$g_i(h) < g_{i-1}(h)$} then
%                   \State $i^*$ = $i-1$ 
%                   \State \textbf{break}
%                \EndIf
%                \If{i=M}
%                   \State $i^*$ = $M$ \\
%                \EndIf   
%         \EndFor           
        \While {$i<M$ and $g_i(h) < g_{i+1}(h)$ } $i=i+1$
	\EndWhile
	\State $\mathcal A_{\text{max}}(h) = \{n_1,\ldots,n_i\}$
%                      
%        
%        
%        
%        \State ${n_2} = r - q$
%        \State Let $L[1 \ldots {n_1} + 1]$ and $R[1 \ldots {n_2} + 1]$ be new arrays
%        \For{$i = 1$ to ${n_1}$}
%            \State $L[i] = A[p + i - 1]$
%        \EndFor
%        \For{$j = 1$ to ${n_2}$}
%            \State $R[i] = A[q + j]$
%        \EndFor
%        \State $L[{n_1} + 1] =  \infty $
%        \State $R[{n_2} + 1] =  \infty $
%        \State $i = 1$
%        \State $j = 1$
%        \For{$k = p$ to $r$}
%            \If {$L[i] < R[j]$}
%                \State $A[k] = L[i]$
%                \State $i = i + 1$
%            \ElsIf {$L[i] > R[j]$}
%                \State $A[k] = R[j]$
%                \State $j = j + 1$
%            \Else
%                \State $A[k] = - \infty$ \Comment{We mark the duplicates with the largest negative integer}
%                \State $j = j + 1$
%            \EndIf
%        \EndFor
%       \EndFunction

\end{algorithmic}
\end{algorithm}

 %%%%%%%%%%%%%%%%%%%%%%%%%%%%%%%%%%%%%%%%%%%%%%%%%%%%%%%%%%%%%%%%%%%%%%%%%%%%%%%%%%%

The time complexity of the above algorithm is  $O(M\log M)$ since time complexity of sorting channel power sequence is $O(M\log M)$ and the time complexity of finding $i^*$ is $O(M)$.
\section{Numerical Results}
In this section we conduct simulations to illustrate the impact of harvesting on the achievable rate. Throughout the section, we take the {\em mean} path loss between the transmitter and the receiver as $-60$dB, i.e., $\mathbb E[H_i] = 10^{-6}$ for all $i\in\mathcal M$. In some of our simulations, we impose a limit the maximum number of antennas that recycle energy at a given time in order to decrease the running time. We denote the maximum number of harvesting antennas as $M_{\text{harvest}}$. We also define {\bf antenna penalty} as the number of antennas that we need to overcome not using recycling. For example, if the antenna penalty is $3$, it means the same rate as the non-recycling system can be achieved with $3$ less antennas with the recycling system.

In Figure~\ref{SNR_compare}, we plot rate $R_{ryc}$ achieved recycling MISO communication provided in Theorem~\ref{lemma:mainresult} and the capacity of non-recycling MISO communication in~(\ref{eq:no-harvest-rate}) as a function of $M$ for different received signal to noise ratio (SNR) values. We set $M_{\text{harvest}}=5$ and all $\{\alpha_{i,j}\}_{i,j\in\mathcal M}$ to $-15$dB. Even though at most $5$ antennas are allowed to harvest, we observe that the rate gap between the harvesting and non-harvesting MISO communication widens as $M$ increases. The reason is that the number of options the transmitter has for selecting the harvesting antennas increases with $M$. {\em The increased spatial diversity in the main channel provides an increased diversity in the choice of recycling antennas.} 

In Figure~\ref{SNR_compare}, we observe rate gains of $\%7$, $\%10$ and $\%20$ associated with recycling at SNR values $10$dB, $0$dB, and $-10$dB, respectively for a system size of $M=25$. Furthermore, the antenna penalty associated with the non-recycling system is larger than $3$, across all SNR values.
%\begin{figure}[t]
%   \centering
%   \includegraphics[width=0.50\textwidth]{rate_change_in_antenna_number_-10dB}
%   % where an .eps filename suffix will be assumed under latex,
%   % and a .pdf suffix will be assumed for pdflatex
%   \caption{}
%   \label{rate_change_minus10dB}
% \end{figure}
%
\begin{figure}[t]
   \centering
   \includegraphics[width=0.40\textwidth]{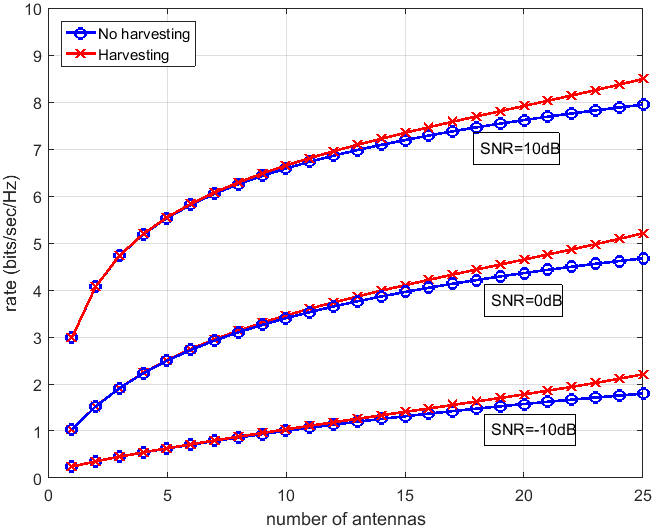}
   \caption{The variation of the achievable rate and the capacity at recycling and non-recycling MISO communications, respectively at different SNR values.}
  \label{SNR_compare}
 \end{figure}

In Figure~\ref{fig:dim_return}, we plot the rate achieved with recycling MISO communication as a function of $M_{\text{harvest}}$. We observe that achievable rate remains constant after $M_{\text{harvest}}$ exceeds $6$ antennas. The reason can be seen in Figure~\ref{fig:average_active}. In Figure~\ref{fig:average_active}, we plot the average number of harvesting antennas as a function of $M$. We observe that the average number of harvesting antennas is smaller than $6$ antennas for any $M$. Hence, picking $M_{\text{harvest}}$ greater than $6$ antennas will not bring an advantage.

\begin{figure}[t]
   \centering
   \includegraphics[width=0.350\textwidth]{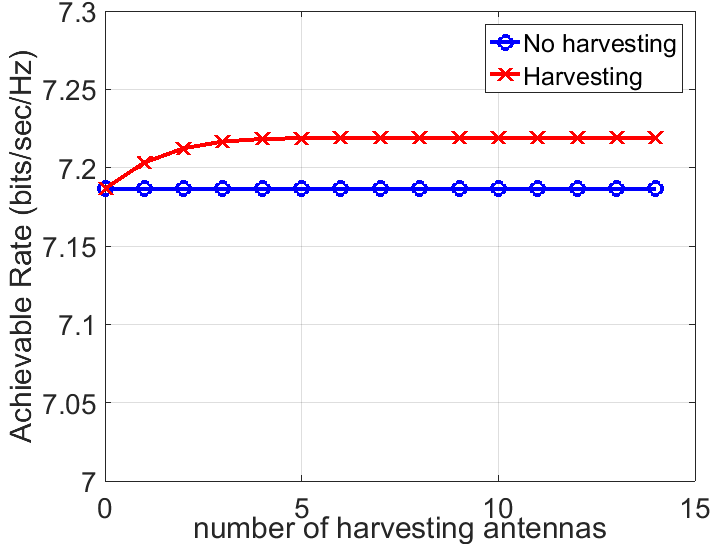}
   \caption{The variation of the rate achieved in recycling MISO communication as a function of number of maximum allowable harvesting antennas.}
   \label{fig:dim_return}
 \end{figure}
\begin{figure}[t]
   \centering
   \includegraphics[width=0.350\textwidth]{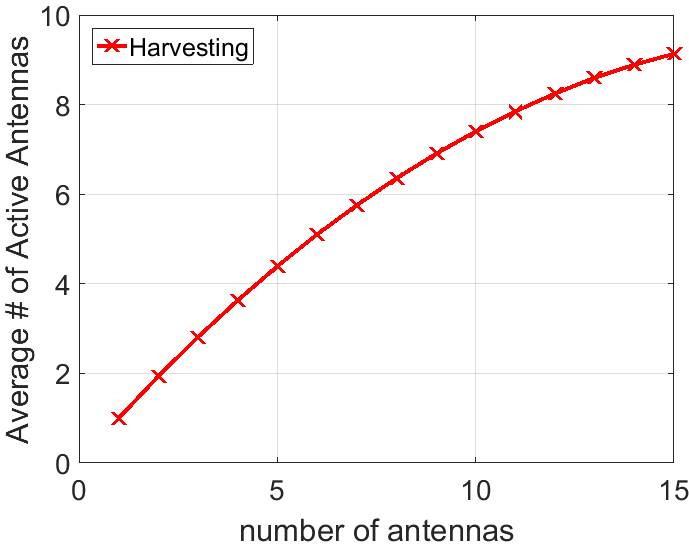}
   \caption{The variation of the average number of active (non-harvesting) antennas as a function of the total number of antennas at the transmitter}
   \label{fig:average_active}
 \end{figure}

We next consider a particular geometry for placing antennas to the transmitter, as shown in Figure~\ref{fig:layout}. Specifically, the antennas are placed on the center and the corners of each small hexagonal. The minimum distance between antennas is $\frac{\lambda}{3}$ that leads to $\alpha=10.3$dB (as per our hardware measurements as discussed in the next section), where $\lambda$ is the wavelength. We set $M_{\text{harvest}}$ to $6$ antennas and the received SNR to $10$dB. In Figure~\ref{fig:hexrate}, we observe that antenna penalty associated with non-recycled transmitter is $4$ antennas at $M=20$.

\begin{figure}[t]
   \centering
   \includegraphics[width=0.40\textwidth]{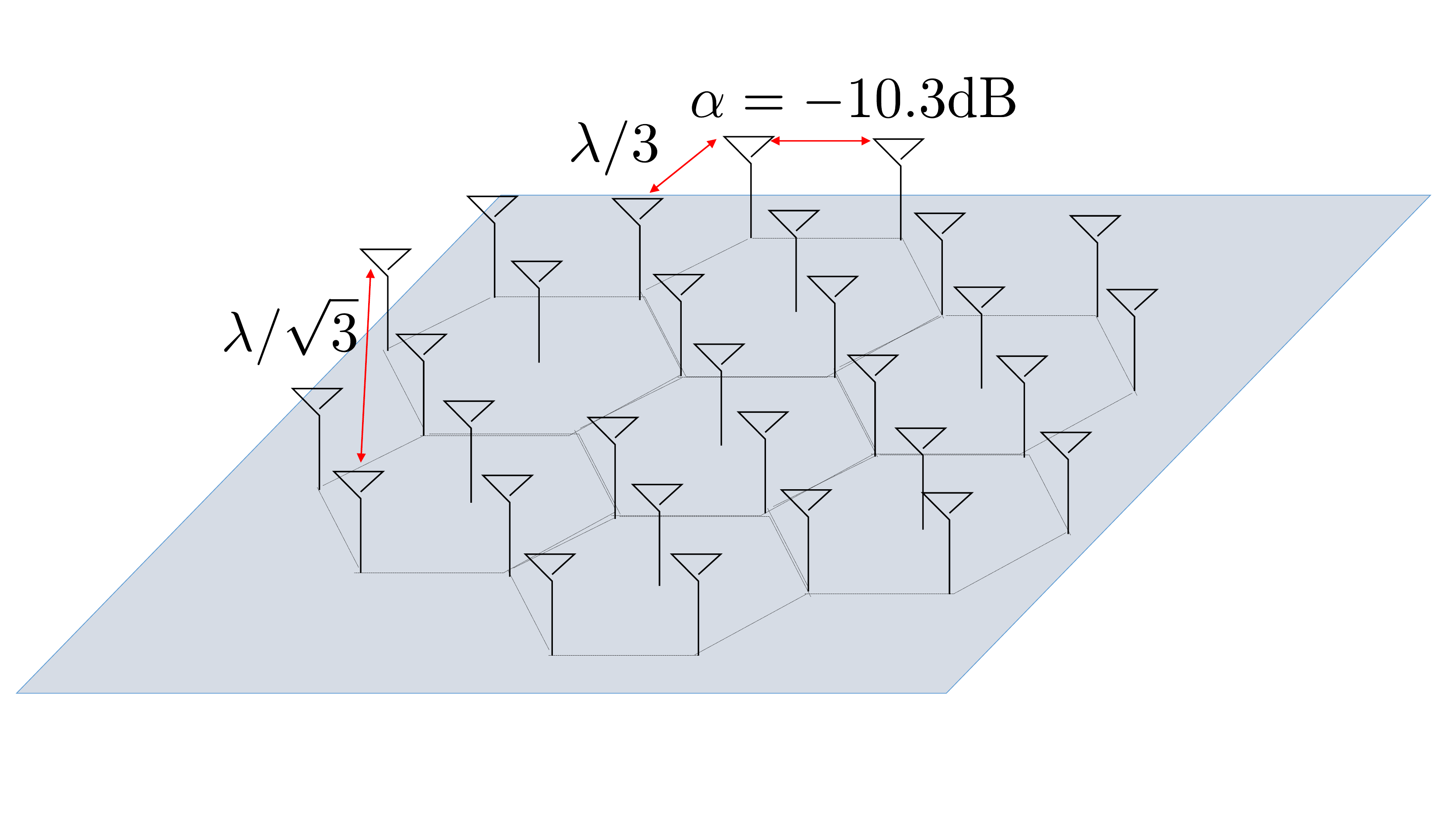}
   \caption{The placement of the transmitter antennas on a hexagonal grid. Each side of small hexagonals has a length of $\frac{\lambda}{3}$, where $\lambda$ is the wavelength.}
   \label{fig:layout}
 \end{figure}
\begin{figure}[t]
   \centering
   \includegraphics[width=0.40\textwidth]{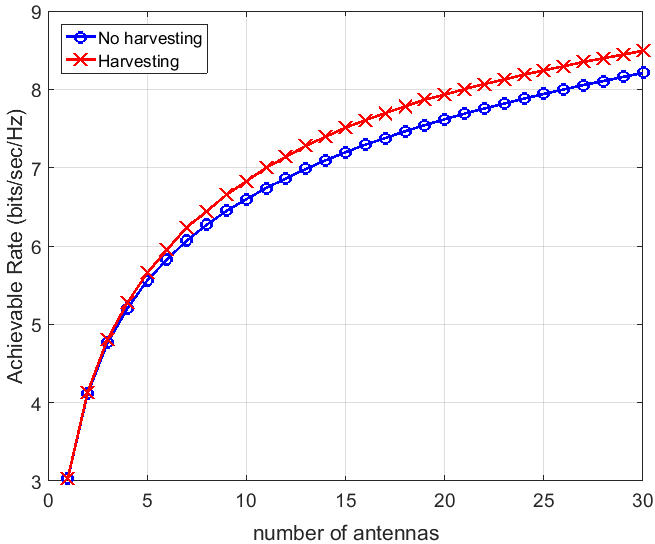}
   \caption{The variation of the achievable rate and the capacity at recycling and non-recycling MISO communications, respectively as a function of the total number of antennas at the transmitter.}
   \label{fig:hexrate}
 \end{figure}

%5) We now consider a particular geometry for placing antennas on a surface. Suppose that the antennas are placed on a hexagonall grid. The minimum distance between antennas  is $0.31\lambda (\alpha=-10dB)$. We set $M_{har}$ to 6 and received SNR to $10dB$. The change of the rate with antenna number can be observed in Figure~\ref{hex_geometry}
%\begin{figure}[t]
%   \centering
%   \includegraphics[width=0.50\textwidth]{hex_geometry_rate_change_in_antenna_number10dB}
%   % where an .eps filename suffix will be assumed under latex,
%   % and a .pdf suffix will be assumed for pdflatex
%   \caption{}
%   \label{hex_geometry}
% \end{figure}

\section{Hardware Measurements and Experimental Results}
\label{sec:exp}

This section aims to address the questions around practicality of energy recycling. In particular, while having an inactive subset of the transmitter antenna array dedicated for energy harvesting, we study the trade-off between harvested and received power at the target receiver. While it can be expected that the closer the harvesting to the active antenna, the larger the power that can be harvested. Meanwhile, it is unclear how does this small separation affect the signal received at the target receiver affected by the coupling between transmission and harvesting antennas. Also, we aim to provide experimental results that shed the light on the selection of active and harvesting antennas. To that end, we use universal serial radio peripherals (USRP) to implement a MISO system with four elements uniform linear antenna array (ULA). First, we provide the hardware setup and the specification of the communication waveform used in these experiments. Then, we provide the experimental results for different selections of active and harvesting array subsets as a function of the array spacing.

\subsection{Hardware Setup}
Hardware system components are listed in the Table \ref{table:hardware_list}. Two CBX RF daughter boards from Ettus Research are installed in a single X300 USRP. Each RF daughter board has a single transmit single receive channels yielding a total of dual transmit dual receive channels. Our experiment goes by activating a maximum of two out of the four element antenna array while the remaining elements are dedicated for energy harvesting. The harvested and received power are measured using a stand alone spectrum analyzer. Figures~\ref{fig:EXP_SETUP} and \ref{fig:spectrum_analyzer} show the transmitter and spectrum analyzer used in our experiment. 
\begin{figure}[t]
   \centering
   \includegraphics[width=0.40\textwidth]{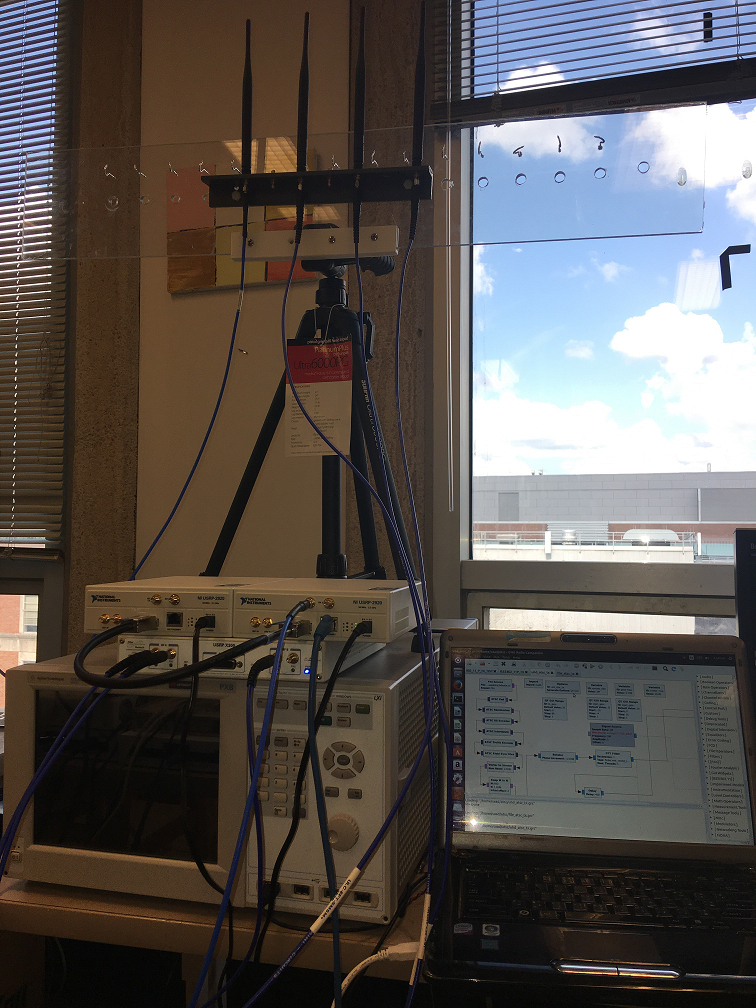}
   \caption{Ettus $X300$ with CBX daughter board takes the role of transmitter}
   \label{fig:EXP_SETUP}
 \end{figure}

\begin{figure}[t]
   \centering
   \includegraphics[width=0.40\textwidth]{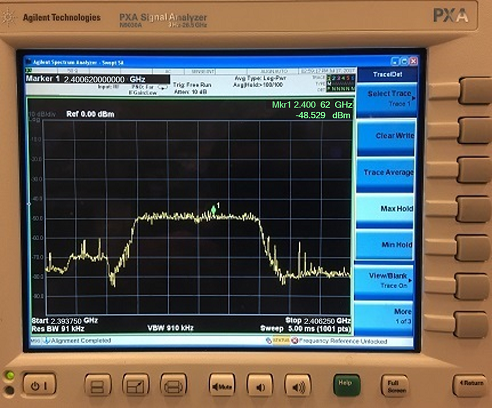}
   \caption{Spectrum analyzer used for power measurements. Center frequency $2.4 GHz$, Signal bandwidth $5 MHz$ , Resolution bandwidth $91 KHz$.}
   \label{fig:spectrum_analyzer}
 \end{figure}

%\begin{figure}[!htbp]
%\centering
%\subfigure[Ettus $X300$ with CBX daughter board takes the role of transmitter]{\includegraphics[width=2.2 in, height =2.5 in]{TX_Harvesting.jpeg} \label{fig:sdr_tx}}
%\subfigure[Spectrum analyzer used for power measurements. Center frequency $1.9272 GHz$, Signal bandwidth $5 MHz$ , Resolution bandwidth $91 KHz$.]{\includegraphics[width=3 in, height = 2.2 in]{Experimental_Results/Spectrum.jpeg} \label{fig:sdr_rx} }
%\label{fig:EXP_SETUP}
%\caption{Experiment Setup.}
%\end{figure}
 
\begin{table*}
\caption{List of Required Hardware Component.}
\label{table:hardware_list}
\begin{center}
 \begin{tabular}{||p{4cm} |p{4cm} |p{8cm} ||} 
 \hline\hline
 \rowcolor{Gray}
 Component & Type & Role in Experiment  \\ [0.5ex] 
 \hline\hline
 CPU & Intel Core i5-3200 CPU 3.40GHz $\times$ 1 & Host for signal processing  \\ 
 \hline
 Operating System & Ubuntu 16.04 LTS, 64 bits & --- \\
 \hline
 GNU Radio & Version 3.7.10 & Signal Processing Environment  \\
 \hline
 USRP & Ettus X300 $\times$ 1  & Transmitter \\
  \hline 
 RF Daughter Board & Ettus CBX $\times$ 2 &  Installed  in one of the X300 USRP to form a dual channel Transmitter  \\ [1ex] 
 \hline
 Spectrum Analyzer & Agilent PXA  & Placed 5 meters away from the transmitter  \\ [1ex] 
 \hline
\end{tabular}
\end{center}
\end{table*}
\subsection{Communication Setup}
In our experimental setup, we have used IEEE-802.11p \cite{s80211std} signal waveform. We have used a slightly modified version of the GNU Radio based implementation of IEEE-802.11p provide in open source by the authors of \cite{bloessl2013ieee}. Table~\ref{table:parameter_list} illustrates typical parameter values used in this experiment.

\begin{table}
\caption{List of typical parameter values used in our experiment.}
\label{table:parameter_list}
\begin{center}
 \begin{tabular}{||c |c||} 
 \hline
 \rowcolor{Gray}
 Parameter & Typical Value \\ 
 \hline\hline
 Center Frequency & 2.4 GHz  \\ 
 \hline
 Bandwidth & 5 MHz \\
 \hline
 FFT Length & 64 \\
 \hline
 Occupied Subcarrires & 52 \\
 \hline
 Data Subcarriers & 48 \\
 \hline
 Pilot Subcarriers & 4  \\
 \hline
 Packet Size & 1000 Bytes  \\
 \hline
 Packet Interval & 150 ms  \\
 \hline
 Transmit Power & +15 dBm \\ [1ex]
  \hline 
 Array Configuration & ULA  \\ [1ex] 
 \hline
\end{tabular}
\end{center}
\end{table}
 
 \subsection{Results}
 First, we activate a single transmit antenna while keeping the other three antennas inactive. We measure the harvested power at each harvesting antenna as well as the power received by an antenna directly attached to the spectrum analyzer. The transmitter antenna array is $5$ meters away from the the spectrum analyzer. The spectrum analyzer provides the average power level measurements per resolution bandwidth. Denote by $P_{res}$, $W$ and $B$ the average power level per resolution bandwidth, the resolution bandwidth and the signal bandwidth, respectively. The total received power, $P_t$, is evaluated as follows:
 \begin{align}
 P_{t} = 10 \log_{10} \left(\dfrac{B 10^{(P_{res}/10})}{W}\right).
\end{align}
We have tested four different selections of active and harvesting antennas. As shown in Figure \ref{fig:selection}, we first activate the left most antenna while the remaining antennas are kept inactive for the sake of energy harvesting. Then, we used the second from the left antenna for data transmission. We then activate the two left most antennas and, finally, we use an interleaved pattern of harvesting and active antennas. At each setup, we measure the total harvested power and the power received at the target receiver for different array separations.

\begin{figure}[t]
   \centering
   \includegraphics[width=0.250\textwidth]{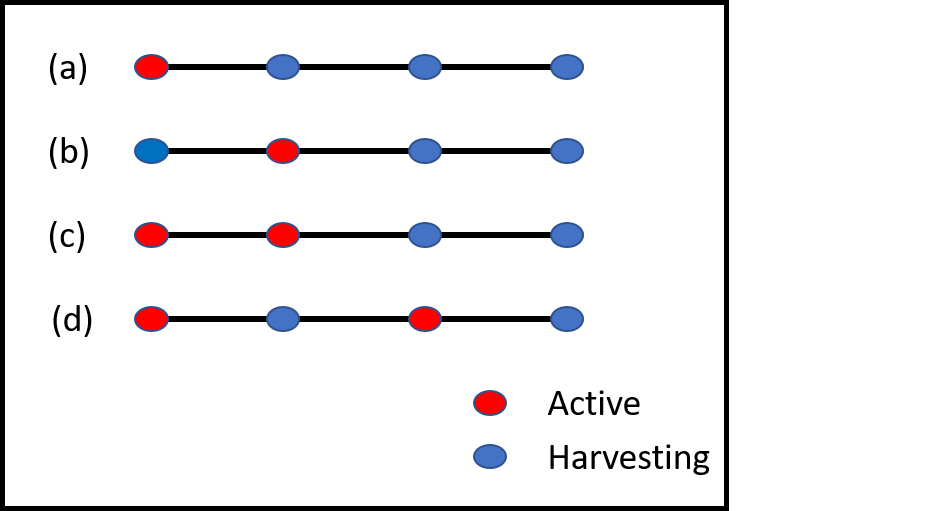}
   \caption{Illustration of our four selections of active and harvesting antennas.}
   \label{fig:selection}
 \end{figure}
   
We started by setting the array spacing equal to $\lambda/4$. In Figure \ref{fig:first}, only the left most antenna is active while the remaining three antennas are dedicated for energy harvesting. We notice that, $\lambda/4$ element spacing yield the best harvesting performance. Meanwhile, the maximum received power is observed for array spacing of $\lambda/2$. Similar results can observed when the second left most antenna is used for transmission as can be seen in Figure \ref{fig:second}. These results suggest that, the closer the harvesting antenna to the transmission antenna, the larger the energy that can be harvested, but also, the larger the coupling between transmission antennas and, hence, the less the power at the target receiver. It worth mentioning that, for array spacing less than $\lambda/4$, while a significant gain of harvesting energy was observed ($4$-$7$dBm), the average received power drops by $3$-$6$dBm due to coupling between the harvesting and transmission antennas.  
 
 \begin{figure}[t]
   \centering
   \includegraphics[width=0.50\textwidth]{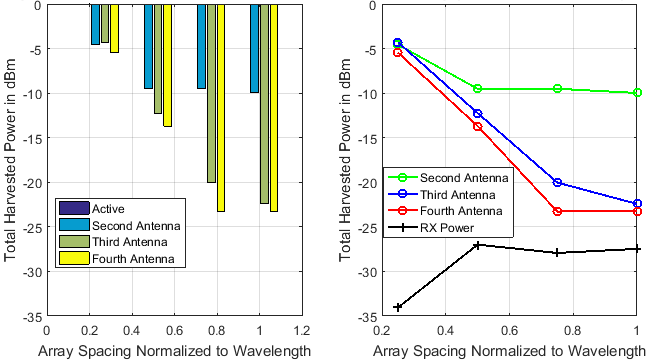}
   \caption{Total harvested and received power with the left most antenna active while the remaining elements are used for harvesting as a function of array spacing.}
   \label{fig:first}
 \end{figure}

\begin{figure}[t]
   \centering
   \includegraphics[width=0.50\textwidth]{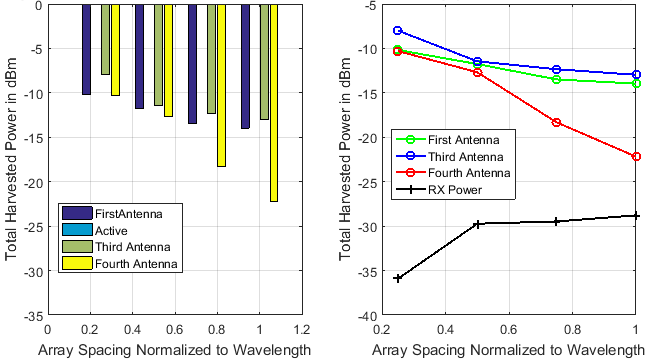}
   \caption{Total harvested and received power with the only the second antenna active while the remaining elements are used for harvesting as a function of array spacing.}
   \label{fig:second}
 \end{figure}
 In Figure \ref{fig:first2}, the two left most antennas are active. Again, we notice that, $\lambda/4$ element spacing yield the best harvesting performance. Meanwhile, the maximum is received power is, again, observed for array spacing of $\lambda/2$. On the other hand, in Figure \ref{fig:interleaved}, the two active antennas are chosen in an interleaved manner. In this setup, we notice that, $\lambda/4$ element spacing provide the best harvesting performance with relatively small reduction in the received power. The gain in harvesting energy is due to having the harvesting antennas closer to two active antennas while maintaining the spacing between the two active elements as $\lambda/2$. This gain comes at the expense of small reduction of received power due to coupling between harvesting and transmission antennas. Meanwhile, for larger array spacing, both the harvested and received power are reduced. This is due to the grating loops generated when the two active elements are separated. We see that the antenna that lies between two transmission antennas was able to maintain better harvesting energy. 
 
The results of this section show that for the special case of a transmitter with ULA, the interleaved configuration with $\lambda/4$ element spacing provides significant gain in harvested energy and, at the same time, maintains reasonably small coupling loss. However, obviously in a {\bf rich scattering} environment, any antenna spacing lower than $\lambda/2$ in the linear arrangement will lead to a reduction in the achieved spatial diversity. Thus, {\em despite striking the optimal balance between recycled and received power, the small antenna spacing may lead to a decrease in the achieved rate in rich scattering environments.}

 %%%%%%%%%%%%%%%%%%%%%%%%%%%%%%%%%%%%%%%%%%%%%%%%%%%%%%%%%%%
 \begin{figure}[t]
   \centering
   \includegraphics[width=0.50\textwidth]{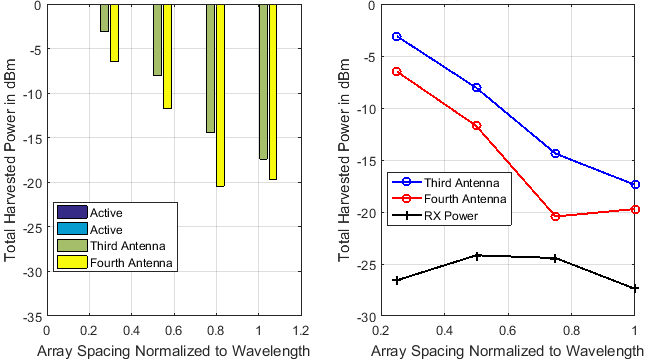}
   \caption{Total harvested and received power with the the two left most antennas are active while the remaining elements are used for harvesting as a function of array spacing.}
   \label{fig:first2}
 \end{figure}

\begin{figure}[t]
   \centering
   \includegraphics[width=0.50\textwidth]{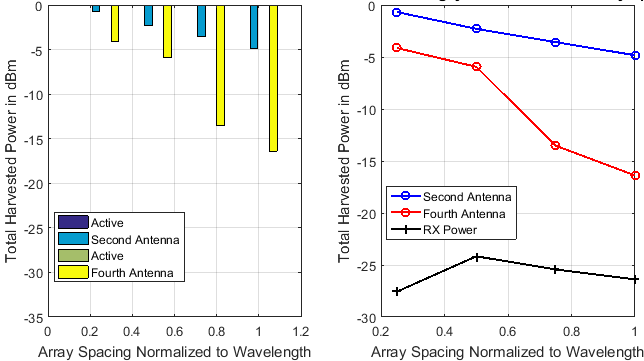}
   \caption{Total harvested and received power with the two active antennas selected in an interleaved manner  while the remaining elements are used for harvesting as a function of array spacing.}
   \label{fig:interleaved}
 \end{figure}
\section{Conclusion}
We studied the basic limits of an energy recycling MISO communication system and showed that it can be utilized to increase the rate achieved, subject to an average power constraint. In our recycling setup, each transmission antenna can alternate between transmission and harvesting states, as chosen by the system. In this setting, we developed the capacity expression for the achievable rate with recycling. We developed the optimal antenna scheduling algorithm that decides on the active and recycling antennas. We also evaluated the optimal power allocation for the active antennas. However, we showed that the optimal scheme has a complexity that grows exponentially with the number of transmission antennas. To address this problem, we developed a near-optimal $O(M \log M)$ complexity algorithm, which preserves optimality in the case of symmetric gains between transmission-harvesting antenna pairs. We numerically evaluated the performance with energy recycling and showed that, with energy recycling, the number of antennas needed to achieve the same rate can be reduced by $10$-$15$\% in typical cases.

To study the issues related to antenna coupling and understand the amount of reduction in received power, we developed a hardware setup and obtained experimental results for a $4\times 1$ MISO system with ULA at the transmitter. Our results reveal that, the closer the harvesting antenna to an active antenna, the larger the harvested power but also the larger the coupling between transmitting antennas, lowering the received power. For that setup, our evaluations showed that, having the transmitting and harvested antennas interleaved (one active followed by one harvesting antenna) at an array spacing equals to $\lambda/4$ yield the best harvested and received power results.

%\section*{Acknowledgment}




\bibliographystyle{IEEEtran}
\bibliography{IEEEabrv,References}
%\begin{thebibliography}{1}

%\end{thebibliography}

\appendices
\section{Proof of Theorem~\ref{lemma:mainresult}}
\label{sec:lemma_proof}
The transmitter employs Gaussian codebook and  conjugate beamforming. The transmitted signal on active antennas can be written as
\begin{align}
X = S \frac{G_\mathcal A^*}{||G_\mathcal A||}\nonumber
\end{align}
where $G_\mathcal{A}\eqdef \{G_k\}_{k\in \mathcal A(H)}$. Note that $||G_\mathcal A||^2 = \sum_{k\in\mathcal{A}(H)}H_k$. Furthermore, $S$ is complex Gaussian signal codeword signal at a particular channel use and is distributed as $CN(0, P(H))$
 
The rate for a given power allocation $P(\cdot)$ is mutual information $I(X;Y|H)$. We skip the derivation of equality 
\begin{align}
I(X;Y|H) = \mathbb E\left[\log \left(1 + P(H)\sum_{k\in \mathcal A(H)}H_k\right)\right]\nonumber
\end{align}
as the derivation is identical with the derivation of the capacity of the non-harvesting MISO communication proof of which can be found in~\cite{tse2005fundamentals}

In the rest of the proof we provide the derivation of the average consumed power at the transmitter. Note that average consumed power is formulated as $
\mathbb E[P(H)] -\mathbb E\left[\sum_{l\in \mathcal{A}^c(H)} |Z_l|^2\right]$, where $\mathbb E\left[\sum_{l\in \mathcal{A}^c(H)} |Z_l|^2\right]$  is the average harvested power. Next we derive the harvested energy as
\begin{align}
&\mathbb E\left[\sum_{l\in \mathcal{A}^c(H)} |Z_l|^2\right]= \nonumber \\
&\mathbb E\left[\sum_{l \in \mathcal A^c(H)} \sum_{m,n\in \mathcal A(H)}  \sqrt{\alpha_{ml}}\sqrt{\alpha_{nl}} X_m X_n^*\right]\nonumber\\
&\qquad\qquad\qquad\qquad  + \mathbb E\left[\sum_l |N_l|^2\right] \nonumber \\
& = \mathbb E\left[|S|^2 \sum_{l \in \mathcal A^c(H)} \sum_{m,n\in \mathcal A(H)} \sqrt{\alpha_{ml}}\sqrt{\alpha_{nl}} \frac{G_m G_n^*}{||G_\mathcal{A}||^2}\right]\nonumber\\
&\qquad\qquad\qquad\qquad + \mathbb E\left[\sum_l |N_l|^2\right]\nonumber\\
&=\mathbb E\left[|S|^2 \sum_{l \in \mathcal A^c(H)}\sum_{m \in\mathcal A(H)}  \alpha_{ml} \frac{H_m}{||G_\mathcal{A}||^2}\right]\nonumber\\
&+ \mathbb E\left[|S|^2 \sum_{l \in \mathcal A^c(H)} \sum_{m,n\in \mathcal A(H):m\neq n} \sqrt{\alpha_{ml}}\sqrt{\alpha_{nl}} \frac{G_m G_n^*}{||G_\mathcal{A}||^2}\right]\nonumber\\
&\qquad\qquad\qquad\qquad + \mathbb E\left[\sum_l |N_l|^2\right] \label{eq:long_format}
\end{align}

We next show that the second term in the RHS of~\eqref{eq:long_format} is equal to zero. First, expand $G_m$ and $G_n$ as $\sqrt{H_m} \exp(i\Phi_m)$ and $\sqrt{H_n} \exp(i\Phi_n)$, respectively, where $\Phi_n$ and $\Phi_m$ are uniformly distributed on $[0, 2\pi]$. Then, we evaluate the second term of \eqref{eq:long_format} as 
\begin{align}
&\text{Second term of RHS of~\eqref{eq:long_format}}=\nonumber\\
&\mathbb E\left[|S|^2 \sum_{l \in \mathcal A^c(H)} \sum_{m,n\in \mathcal A(H):m\neq n} \sqrt{\alpha_{ml}}\sqrt{\alpha_{nl}}\right. \nonumber\\
&\qquad\quad\qquad\qquad\left.  \exp(i\Phi_m)\exp(i\Phi_n)\frac{\sqrt{H_m} \sqrt{H_n}}{\sum_{k\in\mathcal{A}(H)}H_k}\right]\nonumber\\
&=\mathbb E_{S,H}\left[|S|^2 \sum_{l \in \mathcal A^c(H)} \sum_{m,n\in \mathcal A(H):m\neq n} \sqrt{\alpha_{ml}}\sqrt{\alpha_{nl}}\right. \nonumber\\
&\qquad\left. \frac{\sqrt{H_m} \sqrt{H_n}}{\sum_{k\in\mathcal{A}(H)}H_k}\mathbb E_{\Phi_n \Phi_m|H,S}\left[\exp(i\Phi_m)\exp(i\Phi_n)\right]\right]\label{eq:inner_outer}\\
&=0 \label{eq:null},
\end{align}
where the outer expectation in~\eqref{eq:inner_outer} is over $H$ and $S$. The inner expectation in~\eqref{eq:inner_outer} is over $\Phi_m$ and $\Phi_n$ and is conditioned on $(H,S)$. The equality in~\eqref{eq:null} follows due to 
\begin{align}
&\mathbb E_{\Phi_n, \Phi_m|H,S}\left[\exp(i\Phi_n)\exp(i\Phi_m)\right] \nonumber\\
&=\mathbb E_{\Phi_n, \Phi_m}\left[\exp(i\Phi_n)\exp(i\Phi_m)\right] \nonumber\\
&=\mathbb E\left[\exp(i\Phi_n)\right]\mathbb E\left[\exp(i\Phi_m)\right] \nonumber\\
&=0\nonumber,
\end{align}
where the first equality follows from the fact that $(\Phi_m,\Phi_n)$ are independent with $(S,H)$, the second equality follows from the independence of $\Phi_m$ and $\Phi_n$ and the third equality follows from the facts that $\Phi_m$ and $\Phi_n$ are distributed with $CN(0,1)$ and have zero mean.

Showing the second term on the RHS of \eqref{eq:long_format} is zero, we continue the derivation of the average harvested energy as
\begin{align}
&\text{RHS of \eqref{eq:long_format}}\nonumber\\
&=\mathbb E_H\left[\mathbb E_{S|H}\left[|S|^2\right] \sum_{l \in \mathcal A^c(H)}\sum_{m \in\mathcal A(H)}  \alpha_{ml} \frac{H_m}{\sum_{k\in\mathcal{A}(H)}H_k}\right]\nonumber\\
&\qquad\qquad\qquad\qquad + \mathbb E\left[\sum_l |N_l|^2\right]\nonumber\\
&=\mathbb E\left[P(H) \sum_{l \in \mathcal A^c(H)}\sum_{m \in\mathcal A(H)}  \alpha_{ml} \frac{H_m}{\sum_{k\in\mathcal{A}(H)}H_k}\right]\nonumber\\
&\qquad\qquad\qquad\qquad + \mathbb E\left[\sum_l |N_l|^2\right] \label{eq:short_format}
\end{align}
where the second equality follows from the fact data signal $S$ is distributed with $CN(0, P(H))$.

Finally, we can bound the consumed power at the transmitter $
\mathbb E[P(H)] -\mathbb E\left[\sum_{l\in \mathcal{A}^c(H)} |Z_l|^2\right]$ as
\begin{align}
&\mathbb E[P(H)] -\mathbb E\left[\sum_{l\in \mathcal{A}^c(H)} |Z_l|^2\right]=\nonumber\\
&\mathbb E[P(H)]- \mathbb E\left[P(H) \sum_{l \in \mathcal A^c(H)}\sum_{m \in\mathcal A(H)}  \alpha_{ml} \frac{H_m}{\sum_{k\in\mathcal{A}(H)}H_k}\right]\nonumber\\
&\qquad\qquad\qquad\qquad - \mathbb E\left[\sum_l |N_l|^2\right] \label{eq:short_format2}\\
&\leq E\left[P(H)f(H)\right], \nonumber
\end{align}
where the inequality follows from the definition of  function $f(\cdot)$ and from the fact we remove the last term on the RHS of \eqref{eq:short_format2}.
\end{document}